\documentclass[10pt, conference, letterpaper]{IEEEtran}
\ifCLASSINFOpdf
  \usepackage[pdftex]{graphicx}
  \graphicspath{{figures/}}
\else

  \usepackage[dvips]{graphicx}
  \graphicspath{{figures/}}
\fi
\usepackage{subfigure}

\usepackage{amsfonts}
\usepackage{enumerate}
\usepackage{algorithm}    
\usepackage{algorithmic}  
\usepackage{algorithm} 
\usepackage{algorithmic} 
\usepackage{multirow} 
\usepackage{amsmath} 
\usepackage{xcolor}
\usepackage{diagbox}
\usepackage{booktabs}

\usepackage{amsmath}  
\newtheorem{theorem}{Theorem}  
  
\newtheorem{proof}{Proof}[section]  

\begin{document}
%

\title{Rapido: A Layer2 Payment System for Decentralized Currencies}

\author{\IEEEauthorblockN{Changting Lin\IEEEauthorrefmark{1},
 Ning Ma\IEEEauthorrefmark{2},
 Xun Wang\IEEEauthorrefmark{1},
 Zhenguang Liu\IEEEauthorrefmark{1},
Jianhai Chen\IEEEauthorrefmark{2},
Shouling Ji\IEEEauthorrefmark{2}
}
\IEEEauthorblockA{School of Computer and Information Engineering,
Zhejiang Gongshang University, China \IEEEauthorrefmark{1}\\
Email: \{linchangting, wx, lzg\}@zjgsu.edu.cn
}

\IEEEauthorblockA{College of Computer Science and Technology, 
Zhejiang University, China \IEEEauthorrefmark{2}\\
Email: \{3170101236, chenjh919, sji\}@zju.edu.cn
}

}




%


\maketitle

\begin{abstract}
%

Bitcoin blockchain faces the bitcoin scalability problem, for which bitcoin's blocks contain the transactions on the bitcoin network. The on-chain transaction processing capacity of the bitcoin network is limited by the average block creation time of $10$ minutes and the block size limit. These jointly constrain the network's throughput. The transaction processing capacity maximum is estimated between $3.3$ and $7$ transactions per second (TPS). A Layer2 Network, named Lightning Network, is proposed and activated solutions to address this problem. LN operates on top of the bitcoin network as a cache to allow payments to be affected that are not immediately put on the blockchain. However, it also brings some drawbacks. In this paper, we observe a specific payment issue among current LN, which requires additional claims to blockchain and is time-consuming. We call the issue as \textit{\textbf{shares}} issue. Therefore, we propose Rapido to explicitly address the \textit{\textbf{shares}} issue. Furthermore, a new smart contract, D-HTLC, is equipped with Rapido as the payment protocol. We finally provide a proof of concept implementation and simulation for both Rapido and LN, in which Rapdio not only mitigates the \textit{\textbf{shares}} issue but also mitigates the skewness issue thus is proved to be more applicable for various transactions than LN.


\end{abstract}



%
\IEEEpeerreviewmaketitle



\section{Introduction}
\label{intro}
Bitcoin is firstly issued by a blockchain \cite{btc} and the most widely used valuable decentralized digital currency \cite{cmc}. However, it faces a serious issue, i.e., the scalability problem.
To solve this problem, a Layer2 system named Lightning Network (LN) \cite{lightning} is introduced, which is an off-chain payment solution for performing decentralized digital currency payment on top of a blockchain. Specifically, LN is something like the Virtual Private Network (VPN) which establishes a virtual peer-to-peer connection over the Internet. LN theoretically enables fast and secure mircopayments between participating nodes, which features an overlay system for making payments through a network without delegating custody of funds \cite{lightning}. LN is constituted of numerous payment channels that allow the two-party connect by a payment channel securely maintain and update its own ledger by RSMC (Recoverable Sequence Maturity Contract) \cite{lightning}. 
To leverage the existing payment channels to perform end-to-end payments, the payment channel network is introduced. 
Furthermore, HTLCs (Hashed Timelock Contracts) \cite{lightning} is designed, which enables a payment acrosses two or more payment channels with security. 

Benefiting from LN, any transaction between the two nodes among LN can be performed and updated rapidly (an payment takes about $600$ms \cite{concurrencyandprivacy}) instead of confirming by blocks (a on-chain payment should be confirmed by $6$ blocks and take about $11$ minutes on July 2018 \cite{blockconfirm}). As the scale of LN is growing, some limitations have been highlighted, such as sensitive information leakage issue \cite{concurrencyandprivacy, bolt}, payments concurrency issue \cite{concurrencyandprivacy}, skewness issue \cite{revive}, route scalability issue \cite{sprites, flare} and so on. Meanwhile, several works have studied on these limitations \cite{concurrencyandprivacy, bolt, revive, sprites, flare, tumblebit}. In addition, similar payment systems such as credit networks \cite{stella, ripple} also offer some solutions for these analogous issues \cite{silentwhispers, listentoripple, privacypreserving}.

\begin{figure}
\centering
\includegraphics[scale=0.45]{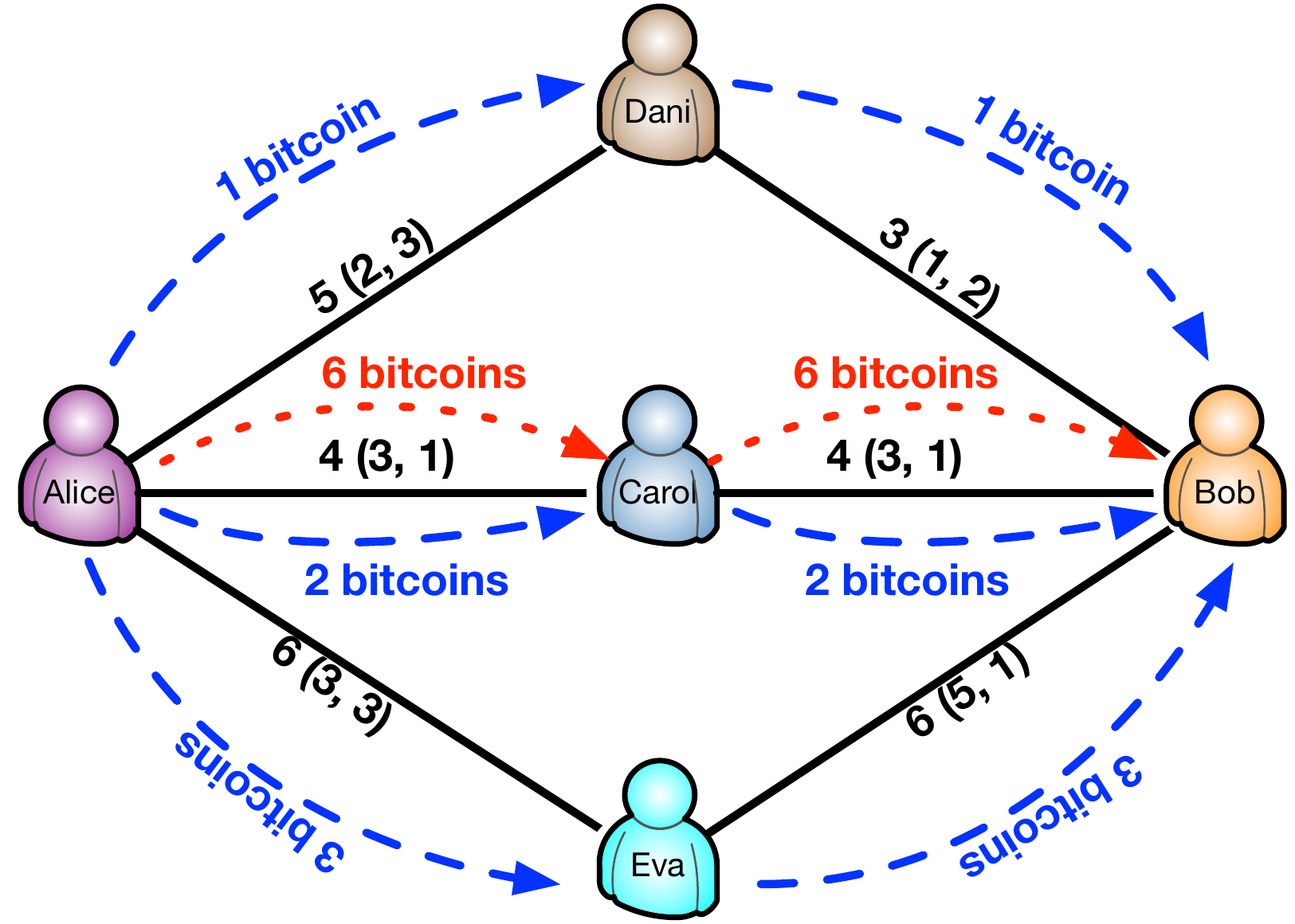}
\caption{\textbf{An Illustrative Example.} The black bold numbers represent the deposits in each payment channel, which is the initial state of the payment channel network. For example, $\textbf{5(2,3)}$ between Alice and Dani represents that this payment channel has total deposits $\textbf{5}$ bitcoins including Alice has $\textbf{2}$ bitcoins and Dani has $\textbf{3}$ bitcoins, respectively. Suppose that Alice is willing to pay Bob $6$ bitcoins and Alice has and only has the $3$ payment channels. One naive way that Alice needs close the other channels and withdraws the deposits and then set up a new payment channel to Bob, which is time-consuming and money-consuming. The blue dash lines represent a simple instance of Rapido to resolve the \textit{\textbf{shares}} issue.}
\label{Fig:example}
\end{figure}

In LN, we observe a payment issue where a node has enough deposits to pay another node while there is no payment channel available to perform this payment. We call this \textit{\textbf{shares}} issue, which appears commonly in the current Layer2 Network. In particular, the \textit{\textbf{shares}} issue is illustrated in Fig. \ref{Fig:example}. Alice is willing to pay Bob $6$ bitcoin while Alice cannot find an available routing path to fulfill this payment. One naive solution is that Alice requires to close the three channels and withdraw all her deposits and then sets up a new payment channel to Bob to fulfill this payment, where the close and set up operations must be claimed to the blockchain for confirming. A node which sets up and closes a payment channel need pay extra fees \cite{lightning}. In addition, the naive solution might also introduce a skewness issue \cite{revive} and payment value privacy issue \cite{concurrencyandprivacy}. Therefore, the \textit{\textbf{shares}} issue not only consumes an extra time and consumes an extra money but also might introduce the other issue such as the skewness and privacy issues.

In this paper, we propose Rapido to address the \textit{\textbf{shares}} issue, which is general and applicable to all payments among the Layer2 Network. Different from LN, a new routing algorithm that incents the nodes always online is proposed in Rapido. Besides, as a transaction is generated, the payment value can be split into several shares by the Value Distributing Problem (VDP) program and then respectively distributed to the other node through these discovered routing paths. With the VDP program, Rapido not only mitigates the skewness of nodes but also preserves the privacy of payment value through splitting the value into shares. In addition, we design D-HTLC (Distributed-HTLC) to guarantee these payments with security. The blue dash lines in Fig. \ref{Fig:example} represent one instance of Rapido to resolve the \textit{\textbf{shares}} issue.



To summarize, the contributions of our paper include the following.
\begin{enumerate}
	\item We observe a payment issue, the \textit{\textbf{shares}} issue, which might appear in the Layer2 Network and causes time-consuming and money-costs.
	\item We design Rapido to mitigate the \textit{\textbf{shares}} issue, which is equipped with D-HTLC and inherently preserves the privacy of total payment value. 
	\item We prove that the VDP is NP-hard.
	\item The simulation demonstrates that Rapido not only mitigates the \textit{\textbf{shares}} issue but also mitigates the skewness of nodes.
\end{enumerate}


\section{Background AND PRELIMINARIES}
\label{background}
In this section, we demonstrate the necessary background of our paper including the blockchain and LN.

\subsection{Blockchain}
A blockchain is a growing list of records, calls blocks, which are linked using cryptography \cite{btc}. Blockchains which serve as the public transaction ledger of cryptocurrencies and are readable by the public of these cryptocurrencies, such as bitcoin \cite{btc} and ethereum \cite{eth}. Each block contains a cryptographic hash of the previous block, a timestamp and transaction data. 
By design, a blockchain is resistant to modification of the data. It is ``an open, distributed ledger that can record transactions between two parties efficiently and in a verifiable and permanent way''. To use as a distributed ledger, a blockchain is typically managed by a peer-to-peer network collectively adhering to a protocol for inter-node communication and validating new blocks. Once recorded, the data in any given block cannot be altered retroactively, which requires a consensus of the network majority.

Bitcoin is a widely used cryptocurrency, a form of electronic cash \cite{btc}. Bitcoins can be paid among the peer-to-peer bitcoin network directly, without the need for intermediaries, though most transactions are made through a cryptocurrency exchange market. Transactions are verified by nodes through cryptography and recorded in a public distributed ledger (blockchain). However, it brings some serious issues, such as the bitcoin scalability problem. In the real world, the on-chain transaction processing capacity maximum is estimated $7$ TPS, which is limited by the average block creation time (about $11$ minutes) and the block size limit ($1$MB) \cite{scalingproblem}.

\subsection{Lightning Network}
To mitigate the bitcoin scalability problem, various solutions are proposed, such as \textit{fork}s (hard \textit{fork} and soft \textit{fork}) \cite{bitcoinsurvey} and the Layer2 systems (LN \cite{lightning} and Sprites \cite{sprites}). \textit{fork} means the protocol changes and divides the blockchain into two distinct entities, which is typically conducted in order to add new features to a blockchain to reverse the effects of hacking or catastrophic bugs \cite{bitcoinsurvey}. Different from \textit{fork}, the Layer2 Network is implemented on top of the blockchain (most commonly bitcoin network) as a cache to allow payments to be performed while does not immediately claim on the blockchain \cite{lightning}. LN is a Layer2 Network, which is constituted of a mass of bi-directional payment channels\footnote{For simplicity, in following payment channel of this paper means bi-directional payment channel.} and theoretically enables fast transactions between participants.

\subsubsection{Payment Channels}

LN, which constitutes of numerous payment channels, is a technique designed to allow two participants to make multiple transactions (e.g., bitcoin) without committing all of the transactions to the blockchain. Each payment channel has two parties, who deposit their own bitcoins and make many secure payments between each other in exchange for making only a few events (e.g., open a new payment channel) on the blockchain. Consequently, the payment channel improves the transaction throughput and eliminates the transfer fee between the two participants. Furthermore, Revocable Sequence Maturity Contract (RSMC) \cite{lightning} must be leveraged one each payment channel, which guarantees the transactions' safety between two-party.

\subsubsection{Payment Channel Network}
A payment channel only permits secure transfer of funds inside a channel while it cannot perform a payment between an indirect-linked two-party. The payment channel network is hence introduced, which is able to perform a secure payment using a series of payment channels by Hashed Timelock Contracts (HTLCs) \cite{lightning}. HTLC leverages hashlocks and timelocks and can allow payments to be securely routed across multiple payment channels without any risk of intermediate nodes stealing the payment in transit.


\begin{figure*}
\centering
\includegraphics[scale=0.35]{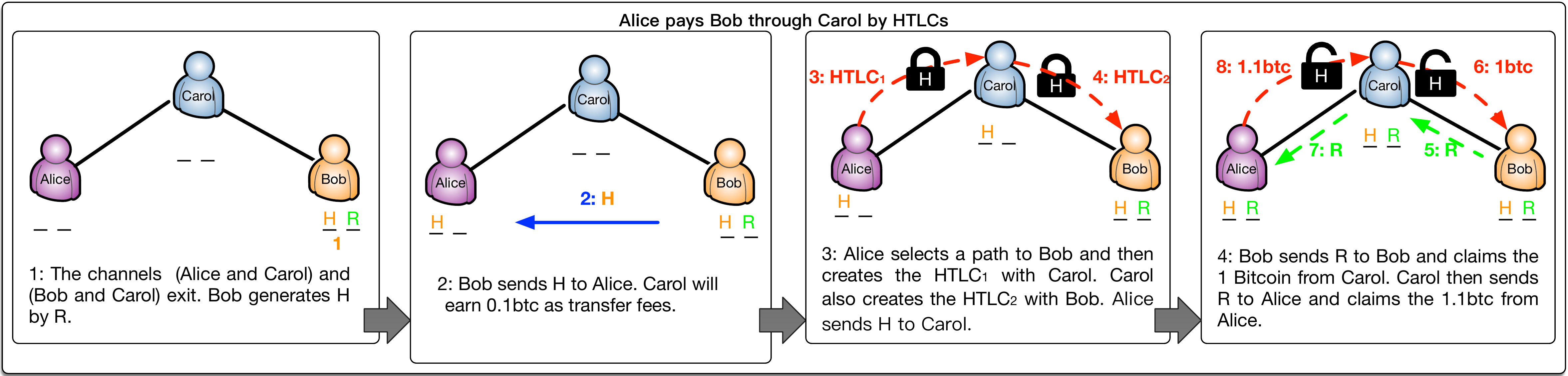}
\caption{\textbf{An Example to Elucidate HTLCs.} Alice wants to pay Bob through Carol by HTLCs. The icons of locks represent the hashtime lock in different status. The blue solid line represents that the $H$ is transferred from Bob to Alice by a secure way; the red dash lines represent all participates have HTLCs and the deposits are locked in payment channels; the green dash lines represent the $R$ is transferred.}
\label{Fig:htlcprocess}

\end{figure*}

In a nutshell, 
suppose that Alice wants to pay Bob $1$ bitcoin through Carol. As an intermediary, Carol earns transfer fees from this payment. The brief steps are shown in Fig. \ref{Fig:htlcprocess} and represented as follows. 
\begin{enumerate}
	\item Bob generates a random number $R$ and calculates its \texttt{SHA-256} hash $H$.
	\item Bob gives $H$ to Alice by secure ways (out of the scope of our paper). Carol will earn $0.1$ bitcoin as transfer fees.
	\item Alice uses her payment channel to pay Carol $1.1$ bitcoin, but she adds a hashlock that Carol gives her to the $1.1$ bitcoin along with an extra condition: in order for Carol to withdraw the $1.1$ bitcoin, she must provide $R$ which was used to produce $H$ in $2$ days\footnote{We use days here as in the original description \cite{lightning}. Instead, the height of blocks is used for description in the real world.}. Carol uses her payment channel to pay Bob $1$ bitcoin, but she also adds a same hashlock that Bob gives her to the payment along with an extra condition: in order for Bob to withdraw the $1$ bitcoin, she must provide $R$ which was used to produce $H$ in $1$ days. 
	\item Bob has $R$ which is used to generate $H$, so Bob can use it to unlock the hashlock and fully receive the $1$ bitcoin from Carol; Carol also uses the $R$ to unlock the hashlock from Alice. Finally, Bob receives $1.1$ bitcoin from Alice. In addition, Carol earns $0.1$ bitcoin as transfer fees.
\end{enumerate}

\subsubsection{Routing}

 In LN, the source should find a routing path to the destination while a payment is generated. The LN routing discovering implements a modified version of Dijkstra's algorithm to find the shortest path between the them \cite{lnd}.
If the available routing path is found by the source, it returns a series of intermediate nodes which encoded the chosen path from the destination to the source.

\section{Motivation and Overview}
\label{motivation}
In the following, we first describe an example to illustrate that the payment performed on the LN causes time-consuming and money-costs. Then, we discuss the motivation of this paper. Lastly, we propose the overview for Rapido.

\subsection{An Illustrative Example and Motivation}

Although LN enables fast payments between participating nodes and has been touted as a solution to the bitcoin scalability problem, drawbacks are also obvious. For example, some payments might not be performed through LN directly. 


Fig. \ref{Fig:example} describes a simple payment case to elucidate the above-mentioned drawbacks. Suppose that the simple topology is a part of LN. Moreover, we suppose that Alice has and only has $3$ routing paths to Bob. However, the naive solution not only cannot resolve the \textit{\textbf{shares}} issue but also brings some issues (privacy, skewness issue, an extra time-consuming and an extra money-costs) , which has been discussed in Section \ref{intro}. These issues endows LN with a poor efficiency, which goes against the original intention of LN.

Some existing works focus on LN \cite{concurrencyandprivacy, revive, sprites}, such as a skewed network issue \cite{revive} and the  payment value's privacy issue \cite{concurrencyandprivacy}. For executing a rebalance process, all participants of rebalance process need to negotiate a better ratio with each other before a rebalance is proceed. A rebalance process might be aborted if the negotiation process fails. Moreover, the participants must be in a loop \cite{revive}. Furthermore, REVIVE cannot resolve the \textit{\textbf{shares}} issue which is shown in Fig. \ref{Fig:example}. For the privacy of payment value, the authors suppose that all intermediate nodes among a payment path are honest and then a trust function is revealed to all of them. By this way, any intermediary of this payment path has the knowledge of the payment value. However, no one can guarantee that all the intermediary are honest \cite{concurrencyandprivacy}.

 We propose an idea to resolve the \textit{\textbf{shares}} issue by multiple path payments. As shown in Fig. \ref{Fig:example}, in an alternative case, Alice splits $6$ bitcoins into $3$ shares, such as $1$, $2$ and $3$, and respectively sends the $3$ shares to Bob through Dani, Carol and Eva, which does not require to close any existing payment channel or set up a new one. We hereby propose a new scheme, named Rapido, that a customer can perform a payment by splitting the payment value into $s$ shares and then respectively performs the $s$ payments to the merchant. Rapido can mitigate the above-mentioned issues, such as \textit{\textbf{shares}}, skewness, payment value privacy time-consuming and money-costs issues, and then perform payments more efficiently among the Layer2 Network.

\subsection{Challenges}
To realize Rapido, there are some challenges to be solved..
\begin{itemize}
	\item How can a customer find available routing paths to the merchant effectively and efficiently.
	\item How to design a split strategy of payment value for customers and avoid serious skewness.   
	\item How to design a smart contract to fulfill the payment effectively and securely.
\end{itemize}

For the first challenge, we implement a routing discovering algorithm with a proactive and a reactive part. For the second challenge, we propose a VDP program. For the third challenge, we design a new smart contract named D-HTLC, which enables the customer to perform a payment successfully.  

\subsection{Overview and Goals}
\subsubsection{\textbf{Overview}}
The overview of Rapdio is shown in Fig. \ref{Fig:process}. What is more, simplified steps of a payment from Client $1$ to Client $2$ through a Layer2 Network are represented as follows.
\begin{enumerate}[\textbf{Step} 1]
	\item Through Beacon Election Module, several nodes are randomly elected as beacons in one period. 
	\item Each Client finds routing paths to all beacons by the Proactive Module.
	\item The information of routing paths (e.g., the deposits distribution among the payment channel) is gathered by Reactive Module of Client 1 when a transaction from Client 1 to Client 2 is generated.
	\item According to the gathered information, Client 1 leverages Value Distributing Module to calculate then select several available routing paths and splits the payment value into several shares.
	\item The Client 1 distributes these shares to Client 2 by D-HTLCs among the selected paths, respectively.  
\end{enumerate}

\begin{figure}
\centering
\includegraphics[scale=0.35]{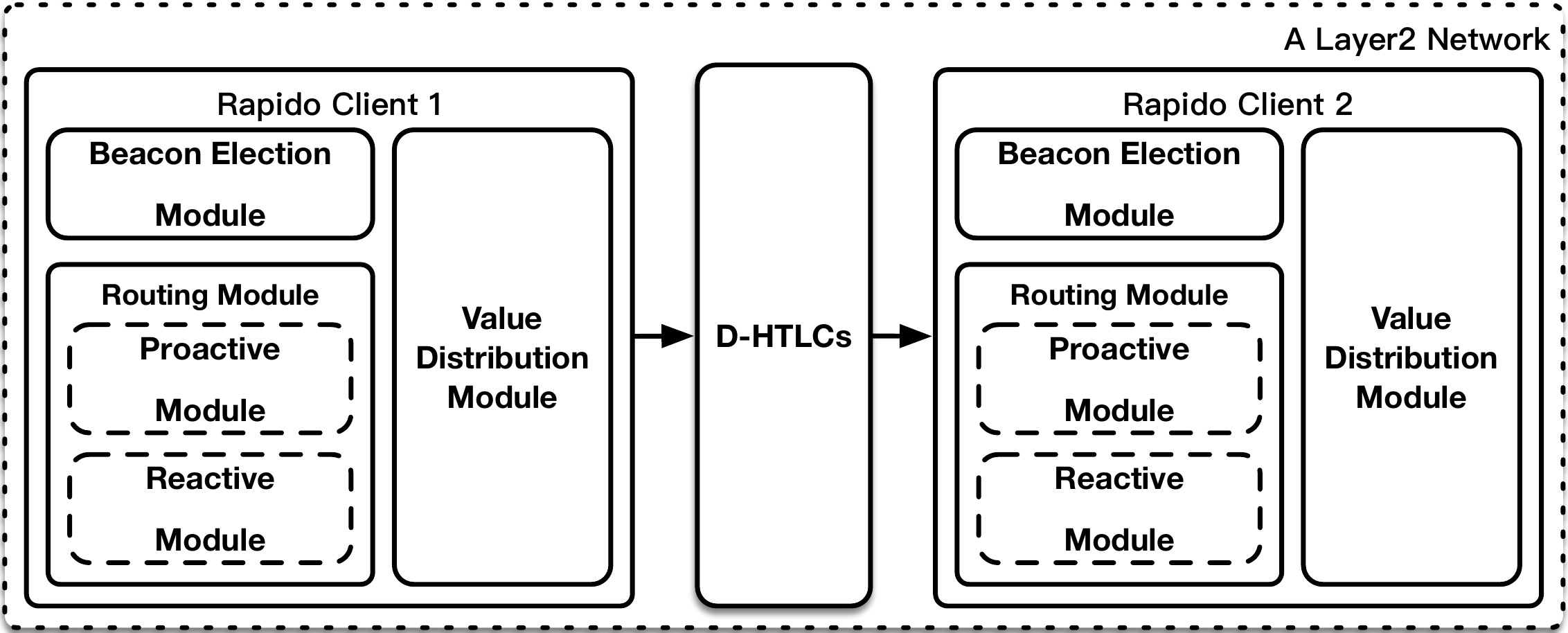}
\caption{\textbf{An Overview of Rapido.} A payment is performed from the Rapido Client $1$ to the Rapido Client $2$ by D-HTLC among a Layer2 Network.}
\label{Fig:process}

\end{figure}

\subsubsection{\textbf{Goals}}
Rapido further reaches the following properties.
\begin{itemize}
	\item \textbf{High-respond.} Rapido mitigates the \textit{\textbf{shares}} issue and reaches a high success rate for numerous mircopayments.
	\item \textbf{Less-skewed.} Rapido mitigates the skewed issue and the congestion issue, such as a skewed node \cite{revive}, after performing lots of payments.
	\item \textbf{Privacy.} Rapido can preserve payment value privacy, which is a private issue of customer and merchant.
	\item \textbf{Generality.} Rather than implementing a solution for $shares$ issue, Rapido is general and applicable to all payments among the Layer2 Network. 
\end{itemize}

\section{The R{apido} Construction}
\label{construction}

In this section, we discuss Rapido construction in detail, which includes \textit{1) Beacons Election, 2) Routing, 3) Value Distributing}, and \textit{4) D-HTLC}. Further, the difference between Rapido and current LN is disscussed in \textit{5) Remark.}

\subsection{Beacon Election}
In Rapido, each node not only can be a normal status but also can be beacon status, which can coexist in one node. Hence, each node has a fixed probability to be elected as a beacon in Rapido. To elect beacons, the network topology is split into $h \in \mathbb{Z}^+$ portions. Then, the nodes in a portion elect a beacon for this portion. The election follows uniform distribution and rotates in a period (e.g., $12$ hours). 
By this way, each node has a chance to be elected as a beacon and can earn transfer fees, which incents all nodes always online as well as improve the entire topology and transaction capabilities of Rapido.


\subsection{Routing in Rapido}
Routing path from a customer to a merchant is prerequisite for a payment to succeed. The routing discovering algorithm implemented in Rapido includes a proactive part as well as a reactive part. The proactive part gathers information that is static or changes slowly, such as payment channels between nodes, whereas the reactive part mainly focuses on dynamic information such as the distribution of deposits among payment channels.



\begin{itemize}
	\item \textbf{Proactive Part.} Number of nodes in a topology are randomly elected as beacons for a pre-defined period. Then, all nodes attempt to find routing paths to all beacons through Breadth-First-Search (BFS) algorithm. Detailed information of routing paths to each beacon is stored in each node's routing table. The routing paths information tends to remain unchanged except a new node joins the Layer2 Network or a new payment channel is set up. In next period, the election process is executed once again. The proactive routing process then is executed repeatedly, and new information is hence stored in each node instead of the elder one.      
	\item \textbf{Reactive Part.} There are some dynamic information in the Layer2 Network, such as the distribution of deposits in each payment channel which consumes lots of computing resources (e.g. CPU) if gathered in real time. In Rapido, instead of gathering all information in real time, the customer and the merchant only gather the dynamic information of each payment channel among the pre-stored routing paths. Moreover, the customer and merchant request the dynamic information of the routing paths to all intermediate nodes. As intermediate nodes, they response the requests by the reactive part of routing. Finally, the gathered information is sent back to the customer. Although the nodes among these routing paths can response the request from the customer and merchant by Tor \cite{lnd, tor}, the Tor achieves a low performance, such as about $5$ seconds to complete $50$kb data request to onion sever \cite{tormetric}. This issue will be discussed in \ref{privacydiss}.
\end{itemize}

\subsection{Value Distributing}
The customer gathers the information of the available routing paths when a transaction is generated, such as the distribution of deposits on these routing paths.
Once information gathering process is finished, the customer can split the payment value into $s$ shares by Linear Program and pay to the merchant through each routing path, respectively. 

In this part, we proceed to formulate the criterion for congestion that represents the state of a node's deposit in a payment channel. Inspired by the congestion factors \cite{factor1,factor2}, we introduce a metric \textit{channel congestion factor} $\mu_{ij} = {{{P_i}} \mathord{\left/
 {\vphantom {{{P_i}} {deposi{t_{ij}}}}} \right.
 \kern-\nulldelimiterspace} {deposi{t_{ij}}}}$, where $P_i$ represents the $i$th share is transferred through $Path_i$ and $deposi{t_{ij}}$ represents the deposit owns to the $j$th intermediate node deposits in $(j-1)$th payment channel over $Path_i$, $i \in [1,s]$ and $j \in [1, l]$. In addition, a \textit{channel congestion factor} can predict the skewness of a node after a payment is performed. According to \cite{factor1} and \cite{factor2}, we introduce another metric \textit{network congestion factor} $\mu = {\max _{{i \in s},{j \in l}}}\mu_{ij}$ which indicates the bottleneck of the whole network. For the value distributing issue, it can be equivalent to the well-known Maximum Flow Problem \cite{mfp} \textit{iff} there are no other restrictions on the routing paths. Therefore, one can find a solution that minimizes the \textit{network congestion factor} in polynomial time through a standard max-flow algorithm \cite{mfp}. However, the cumulative transfer fees in a routing path is one of essential metrics for the customer, which should be considered in the value distributing process.

For this problem, we aim at minimizing the \textit{network congestion factor} $\mu$ which subjects to a restriction on a transfer fees according to the chosen paths. It is formulated as follows.

\textbf{VDP (Value Distributing Problem).} Given a Layer2 Network $\mathbb{G} = \left( {\mathbb{V},\mathbb{E}} \right)$, a source-destination pair ($\mathcal{C}, \mathcal{M} \in \mathbb{V}$), each participant within a transaction sends their transfer fees $fee_v$ information and the deposits $deposi{t_e}$ to $\mathcal{C}$ $(v \in \mathbb{V}, e \in \mathbb{E})$. Besides, a payment has a transfer fee restriction $\rho$. 
Finally, we find a feasible solution to distribute $s$ shares that minimize the \textit{network congestion factor}.

\begin{theorem}
	\textbf{VDP is NP-hard.}
\end{theorem}


\begin{proof}
Consider the following Partition Problem \cite{kpp}. Given a set $A=\left\{ {{\alpha _1},{\alpha _{2,}} \cdots ,{\alpha _{n}}} \right\}$. Moreover, the size of $\alpha_i$ can be represented as $S({\alpha_i}) \in {\mathbb{Z}^+ } $. Then, we can find a subset $A' \subseteq A$ such that $\sum\nolimits_{\alpha  \in A'} {S(\alpha )}  = \sum\nolimits_{\alpha  \in A - A'} {S(\alpha )} $. We transform the Partition Problem to a simple VDP as follows.
\begin{enumerate}
	\item Given an element $\alpha_i \in A$ with size $S(\alpha_i)$. In a Layer2 Network, we suppose a transaction between $\mathcal{C}$ and the $\mathcal{M}$, where $P$ would be transferred and the set $\{ A \} $ constitute all intermediate nodes. We define a transfer fee of node $i$ as $S(\alpha_i)$.
	\item Suppose that $P$ is split into two parts ($P_1$ and $P_2$) for transferring through each routing path, respectively. Suppose that $P_1$ is transferred from $\mathcal{C}$ to $\mathcal{M}$ through a routing path $\{ A' \} $ and $P_2$ is transferred through another routing path $\{ A-A' \} $. Besides, we assume that each channel's capacity is ${\raise0.7ex\hbox{$P $} \!\mathord{\left/
 {\vphantom {P  2}}\right.\kern-\nulldelimiterspace}
\!\lower0.7ex\hbox{$2$}}$. 
\end{enumerate}

We shall prove that it is possible to transfer $P_1$ and $P_2$ through each routing path, that the total of transfer fees are not exceed $\rho$ and each payment channel's \textit{channel congestion factor} $\mu$ is not exceed to $1$ \textit{iff} there exits a subset $A' \subseteq A$ and exits $\sum\nolimits_{\alpha  \in A'} {S(\alpha )}  = \sum\nolimits_{\alpha  \in A - A'} {S(\alpha )}  = {\raise0.7ex\hbox{$\rho $} \!\mathord{\left/
 {\vphantom {\rho  2}}\right.\kern-\nulldelimiterspace}
\!\lower0.7ex\hbox{$2$}}$.

Note that if variables $x_1 , x_2$ satisfy $x_1 \le   {\raise0.7ex\hbox{$\rho $} \!\mathord{\left/
 {\vphantom {\rho  2}}\right.\kern-\nulldelimiterspace}
\!\lower0.7ex\hbox{$2$}}$ and $x_2 \le {\raise0.7ex\hbox{$\rho $} \!\mathord{\left/
 {\vphantom {\rho  2}}\right.\kern-\nulldelimiterspace}
\!\lower0.7ex\hbox{$2$}}$, it follows that $\max \left( {{x_1} + {x_2}} \right) = \rho$ and in addition $x_1 =x_2 = {\raise0.7ex\hbox{$\rho $} \!\mathord{\left/
 {\vphantom {\rho  2}}\right.\kern-\nulldelimiterspace}
\!\lower0.7ex\hbox{$2$}}$. We hence conclude that $\sum\nolimits_{\alpha  \in A'} {S(\alpha )}  \le {\raise0.7ex\hbox{$\rho $} \!\mathord{\left/
 {\vphantom {\rho  2}}\right.\kern-\nulldelimiterspace}
\!\lower0.7ex\hbox{$2$}}$ and $\sum\nolimits_{\alpha  \in A - A'} {S(\alpha )} \le {\raise0.7ex\hbox{$\rho $} \!\mathord{\left/
 {\vphantom {\rho  2}}\right.\kern-\nulldelimiterspace}
\!\lower0.7ex\hbox{$2$}}$, it follows that $\max \left( {\sum\nolimits_{\alpha  \in A'} {S(\alpha )} + \sum\nolimits_{\alpha  \in A - A'} {S(\alpha )}} \right) = \rho$ and in addition $\sum\nolimits_{\alpha  \in A'} {S(\alpha )} = \sum\nolimits_{\alpha  \in A - A'} {S(\alpha )} =  {\raise0.7ex\hbox{$\rho $} \!\mathord{\left/
 {\vphantom {\rho  2}}\right.\kern-\nulldelimiterspace}
\!\lower0.7ex\hbox{$2$}}$. 
In this issue, $P$ can be split into two parts as $P_1= {\raise0.7ex\hbox{$P$} \!\mathord{\left/
 {\vphantom {P 2}}\right.\kern-\nulldelimiterspace}
\!\lower0.7ex\hbox{$2$}}$ and $P_2= {\raise0.7ex\hbox{$P$} \!\mathord{\left/
 {\vphantom {P 2}}\right.\kern-\nulldelimiterspace}
\!\lower0.7ex\hbox{$2$}}$ that the \textit{channel congestion factor} $\mu  \le  1$ is not be violated. We hereby find two part $P_1= P_2= {\raise0.7ex\hbox{$P$} \!\mathord{\left/
 {\vphantom {P 2}}\right.\kern-\nulldelimiterspace}
\!\lower0.7ex\hbox{$2$}}$ that the total fees restrict to $\rho$. The simple VDP hereby can be reduced to a Partition Problem such that the positive integer $S(\alpha_i)$ can be seem as a transfer fee of the node $i$. 

The Partition Problem is known as strongly NP-hard \cite{kpp}. We note that the VDP has more restrictions than the simple VDP such that a payment value $P$ could be split into $s   \ge 2$ shares and transferred through $s$ routing paths from a customer to a merchant after a transaction is generated in the Layer2 Network.
The VDP hereby can be considered as a k-Partition Problem \cite{kpp} which is known as strongly NP-hard \cite{kpp}. Consequently, the \textbf{VDP is NP-hard}.
\end{proof}

Owing to the VDP is NP-hard  without polynomial time algorithm. To resolve VDP, it should be relaxed to a non-integer linear programming to obtain a approximate solution. The first step towards obtaining a solution to Problem VDP is to define it as a linear programming. To that end, we need some additional notations.

\textbf{Program VDP:}
\begin{alignat}{2}
\label{lp}
\mbox{minimize} \quad & \mu  \\
\mbox{s.t.} \quad & \sum\limits_{e \in Out(v)} {f_e^\lambda }  - \sum\limits_{e \in In(v)} {f_e^{\lambda  - Fe{e_v}}}  = 0   \\  
\quad & \sum\limits_{e \in Out(\mathcal{C})} {f_e^\lambda }  - \sum\limits_{e \in In(\mathcal{C})} {f_e^{\lambda}}  = 0  \\
\quad & \sum\limits_{e \in Out(\mathcal{M})} {f_e^\lambda }  - \sum\limits_{e \in In(\mathcal{M})} {f_e^{\lambda}}  = 0  \\
\quad & \sum\limits_{v \in R} {Fe{e_v}}  \le \rho  \\
\quad & P = {P_1} +  \cdots  + {P_s}  \\
\quad & {P_i} < deposi{t_{ij}}  \\
\quad & \mu_{ij}  \ge {\mu _{threshold}}  \\
\quad & f_e^\lambda  \ge 0     \\
\quad & \mu \ge 0  
\end{alignat}

Recall that we are given a network $\mathbb{G} = \left( {\mathbb{V},\mathbb{E}} \right)$, $\left| \mathbb{V} \right| = N$ and $\left| \mathbb{E} \right| = M$, source-destination pair $\mathcal{C}, \mathcal{M} \in \mathbb{V}$, a payment is performed $P$ between $\mathcal{C}$ and $\mathcal{M}$ and go through $R \in N$ nodes, a transfer fee $Fee_v > 0$ for each node $v \in \mathbb{V}$, and a node's deposit in its former channel (i.e., the $j$th node's deposit of the $(j-1)$th channel in $Path_i$) $deposit_{ij}$ for each payment channel $e \in \mathbb{E}$, total of transfer fees restriction $\rho$ for a payment. Let $\mu$ be the \textit{network congestion factor}. Denote by $f_e^\lambda $ the flow along $e=(u,v) \in \mathbb{E}$ that has been routed from $\mathcal{C}$ to $u$ through paths with a total fees $\lambda$. Finally, for each $v \in \mathbb{V}$, denote by $Out(v)$ the set of channels that out from $v$, and by $In(v)$ the set of channels that input node $v$. To mitigate a congestion issue and skewed node issue, a payment is performed \textit{iff} $\mu_{ij}  \ge {\mu _{threshold}}$ where ${\mu _{threshold}}$ is pre-defined. Therefore, the VDP can be formulated as a linear program, as specified in \textbf{Program VDP}.

%

The objective function of \textbf{Program VDP} minimizes the\textit{ network congestion factor} $\mu$. The constraint equations $(2)$, $(3)$ and $(4)$ denote that the conservation constraint of the payments and transfer fees through each node. Equation$(2)$ states that the cumulative fees of a payment out of node $v$ has to be equal to the cumulative fees before this payment inputs node $v$ and the transfer fees for node $v$. For the equations$(2)$ and $(3)$, any payment from a $\mathcal{C}$ to a $\mathcal{M}$, in which the transfer fees do not need to pay a $\mathcal{C}$ or a $\mathcal{M}$. The inequation $(5)$ states that the constraint of a payment's total of transfer fees. Equation$(6)$ states that the sum of the shares must be equal to the payment value. For the inequation $(7)$, the deposits in the $(j-1)$th payment channel owns to the $j$th node which is in the $Path_i$ must larger than the $i$th share. To mitigate the skewness and congestion issue, the inequation $(8)$ states that the \textit{payment channel congestion} $\mu_{ij}$ must be not less than a pre-defined threshold $\mu_{threshold}$. Finally, the inequation $(9)$ and $(10)$ restrict the variables to be nonnegative.

According to the result of Program VDP, the customer can splits the payment value into shares. However, there might exist some intermediate nodes disapprove of this payments. Therefore, we introduce a \texttt{request} mechanism which is executed before distributing these shares. In \texttt{request} mechanism, the customer proposes a proposal and requests each node among these selected routing paths such that, whether participating this payment as an intermediate node. The customer is going to perform this payment \textit{iff} each node agrees this proposal with the customer. Otherwise, the customer aborts these selected routing paths and reduplicate the \textit{Routing} process. The \texttt{request} mechanism can mitigate some extra costs that an intermediate node might abort his share payment when a payment is performing.

\subsection{D-HTLC}

D-HTLC is a smart contract which is based on HTLC \cite{lightning}. In D-HTLC, we introduce a \texttt{punish} mechanism to avoid some traps.For example, nodes on routing paths might maliciously abort before payment, resulting in the reset of D-HTLC and, unavoidably, the waste of extra costs. To solve this problem, we lock a certain amount of $cash$ for each node until a D-HTLC is set up for this payment.

\subsubsection{Notation}
Input a network $\mathbb{G} = \left( {\mathbb{V},\mathbb{E}} \right)$, suppose that a customer $\mathcal{C} \in \mathbb{V}$ attempts to transfer $P$ coins to a merchant $\mathcal{M} \in \mathbb{V}$. Then, $P$ has been spilt into $s$ shares $listP_{share} = \left\{ {{P_1}, \cdots ,{P_i}, \cdots ,{P_s}} \right\}$ by Program VDP, that each share has its own routing path $listPath = \left\{ {{Path_1}, \cdots ,{Path_i}, \cdots ,{Path_s}} \right\}$ to $\mathcal{M}$, such as the $i$th share $P_i$ is paid to $\mathcal{M}$ through $pat{h_i}$. For each routing path, $\mathcal{C}$ generates a $H_i =$\texttt{SHA-256}$(R_i)$. Besides, the $node_{ij}$ means that the $j$th node is on $pat{h_i}$. The D-HTLC $ID_{ij}$ is deployed between $node_{ij}$ and $node_{i(j+1)}$, which is composed of sender $node_{ij}$, receiver $nod{e_{i\left( {j + 1} \right)}}$, the $i$th value $P_i$, transfer fee $fee_{ij}$, hashlock $hLk_i$ and timelock $tLk_j$. In addition, $node_{ij}$ should prepare a $cash_{ij}$ according its own transfer $fee_{ij}$ (base fee and fee rate). We represent $listCash$ as the list of all $cash$s.

\subsubsection{Definition of Operations}

  D-HTLC has $6$ main operations: \texttt{\textbf{openPunish}}, \texttt{\textbf{newContract}}, \texttt{\textbf{getBack}}, \texttt{\textbf{withdraw}}, \texttt{\textbf{refund}} and \texttt{\textbf{punish}}. We briefly describe these operations as follows.

\begin{itemize}
	\item \texttt{\textbf{openPunish}}$\left( {\mathcal{C},listNode,listCash} \right).$ When a client $\mathcal{C}$ attempts to initialize a payment, the \texttt{\textbf{openPunish}} is called. Nodes that responses request from $\mathcal{C}$ are listed in $listNode$. This function is responsible for locking a certain amount of cash which are listed in $listCash$ from every involved nodes. It returns $1$ \textit{iff} all nodes in $listNode$ have locked the corresponding amount of $cash$; otherwise, it returns 0.   
	\item \texttt{\textbf{newContract}}$\left( {I{D_{i\left( {j - 1} \right)}},nod{e_{ij}}, fe{e_{ij}}, hLk_{i},tLk_{ij}} \right).$ This operation constructs a new D-HTLC from  from the $node_{ij}$ to the $nod{e_{i\left( {j + 1} \right)}}$. It returns $ID_{ij}$; otherwise, it returns $0$. Note that sequence between several \texttt{\textbf{newContract}}s should be consist with the sequence of nodes on the path.
	\item \texttt{\textbf{getBack}}$\left( {nod{e_{ij}},cas{h_{ij}}} \right).$ This operation is called when $node_{ij}$ is trying to retrieve its locked $cash_{ij}$ after  \texttt{\textbf{newContract}} succeeds and $ID_{ij}$ is returned.
	\item \texttt{\textbf{withdraw}}$\left( {I{D_{ij}},R_i} \right).$ This operation is called when $nod{e_{i\left( {j + 1} \right)}}$ attempts to withdraw $P_i$ from the previous $node_{ij}$ within $ID_{ij}$. It succeeds \textit{iff} $nod{e_{i\left( {j + 1} \right)}}$ is able to reveal the correct $R_i$ to $node_{ij}$. 
	\item \texttt{\textbf{refund}}$\left( {I{D_{ij}},{P_i}} \right).$ If there was no \texttt{\textbf{withdraw}} until $tLk_{ij}$ is expired, \texttt{\textbf{refund}} is called to refund $P_i$ to $node_{ij}$. 
	\item \texttt{\textbf{punish}}$\left( {\mathcal{C},nod{e_{ij}},cas{h_{ij}}} \right).$ As previously mentioned, we have a punish mechanism to penalize dishonest nodes. Thus, this punish function is called once a new D-HTLC aborts because of $node_{ij}$, then $\mathcal{C}$ would get $cash_{ij}$. For the latter nodes, the remaining locked $cash$s are unlocked.
\end{itemize}

 \subsection{Remark}
 We design a Layer2 system, Rapido, which is different from the current LN. At the first, we design a new routing discovering algorithm for Rapido, which incents nodes always online. Secondly, Rapido splits the payment value into several shares and distributes these shares to a node by D-HTLC. Thirdly, D-HTLC is designed based on HTLC while D-HTLC inherently preserves the privacy of payment value.

\section{System Discussion}
\label{discussion}
In this part, we discuss the scenario where Rapido can be used as well as its limitations..

\subsection{Usability}

\subsubsection{Context}
Employing the Layer2 Network makes transactions more rapid and convenient than the transactions are performed on-chain. As summarized in the above discussion, LN has some issues while Rapido can mitigate them. In particular, for the \textit{\textbf{shares}} issue, the naive solution introduces an extra time and money costs whereas Rapido can eliminate these costs, which is not require to claim in blockchain. 

Different from HTLC, a payment is performed by D-HTLC should be split into several shares and distributed to the merchant through selected routing paths. In other words, a payment is performed by D-HTLC requires more participants. According to the statics of LN from \cite{lne}, $2681$ nodes that each node has $8$ payment channels on average. Hence, the  condition provides a suitable context for Rapido. Moreover, the Layer2 Network such that LN's scale is increasing from its birth to the present. The Layer2 Network with its growth, the future context is also suitable for Rapido.

\subsubsection{Scalability}
With the scale expansion of the Layer2 Network, more and more nodes and channels would be generated. For the proactive part of routing, it is $O\left( n \right)$. Moreover, the beacon election is only executed once in each period, which is not always consumed computing resource. Moreover, the pre-defined period of election is not short. In reactive part, for performing a payment, the customer calculates and selects available routing paths through several beacons and then pay the split payment shares to the merchant. Considering the number of beacons is always much less than the total number of nodes, the process of routing path selection is rapid unless the number of beacons grows clipping.    

\subsection{Security Analysis}

\subsubsection{Threat Model}
We suppose that an adversary would be willing to intercept some information from a transaction between a customer and a merchant among the Layer2 Network. A transaction over the Layer2 Network might go through serval intermediate nodes, whereas security flaws of some nodes can be found and may be leveraged by an adversary. Therefore, a payment value can be intercepted if one of intermediate nodes is compromised. An adversary in our paper targets the privacy of payment value, which is crucial as the other security issues (e.g., anonymity of customer/merchant) when a payment is performing on the Layer2 Network.   

\subsubsection{Guarantees for privacy}
\label{privacydiss}
Under the previous LN discussion and adversarial assumption, an adversary can easily intercept the payment value if one of intermediate nodes is compromised. In Rapido, the adversary cannot intercept the total payment value unless it at least compromises one intermediate node on each routing path. This is because the payment value is split into several shares and distributed to the merchant through numerous paths by the customer while the adversary can only intercept a part of payment value. Moreover, a payment based on D-HTLC includes more routing paths as improving difficulty level to compromise all intermediate nodes. Hence, it is difficult that an adversary could intercept a complete information of payment value. In addition, the Rapido also could confuse the adversary, which only leaks a part of information. Rather than Fulgor and Rayo \cite{concurrencyandprivacy}, Rapido need not introduce an extra trust function to preserve the privacy of payment value and make a assumption that all participates are honest.   

 In addition, any node can requests information of payment channel to perform the payment in reactive part. Although the privacy of information can be transmitted and protected by optional Tor \cite{lnd, tor}, it achieves a low performance \cite{lnd, tormetric}. We note that it might cause some information leakage in Rapido, if the Tor is turnoff. Considering that lots of nodes are elected as beacons, there are a majority of nodes would be as intermediate nodes when a payment is performed. Moreover, lots of payments might be generated on the Layer2 Network in a short time. Therefore, the information, e.g., the node's deposits on each payment channel, is rapidly changing and out date quickly. Consequently, Rpaido theoretically enables information leakage with a low probability.

\section{Simulation}
\label{experiment}

In this section, we present extensive experiments to evaluate Rapido. To be more convincing, we provide a working proof of concept (POC) implemented in Python with a real dataset of LN. 

\subsection{Dataset and Setup}
The experiments were conducted on machines with Intel i$7$ $8700$K and $32$GB RAM. For simulating a Layer2 Network, we crawled a dataset with $2,681$ nodes and $7,347$ payment channels from LN of Recksplorer \cite{lne} on July $2018$. Besides, detailed information of each node and payment channel are also demonstrated, such as capacity, base transfer fee, rate of transfer fee and so on. In common, the $20$\% people always own $80$\% of all wealth \cite{pareto}, which also can be analogous to LN. Based on the fact that there usually exits some super ``huge'' nodes which owns lots of payment channels ($60+$) and thus deposits (e.g., ``tady je slushovo'' \cite{lne}), we assume that the distribution of deposits on each channel follows the Pareto distribution \cite{pareto}. Under this assumption, we suppose that two nodes among a payment channel, where the node who own more channels has $80$\% deposits. We simulate a simple blockchain by golang, which has links open on our Layer2 Network. 
Based on these statics, assumptions and setup, we derive datasets for our simulation which are represented as following.

\subsection{Beacon Election and Average Routing Hops}
 In our work, the routing algorithm includes proactive part and reactive part. In proactive part, some beacons are elected, and then each node finds a routing path to each beacon. The left half of Fig. \ref{Fig:beacon} represents that the average hops between any two nodes among the Layer2 network topology. The label of X-axis ($5$, $50$, $200$ and $500$) represents numbers of beacon nodes. We can find that the $5$, $50$, $200$ and $500$ labeled bars are of about the same height, with the average hops of about $5.1$. However, the more beacons are elected, the more routing information should be stored in routing table each node. To avoid routing table overhead and meanwhile incent more nodes to be online, we empirically elect $200$ beacons in topology.
 
 On the other hand, the right half of Fig. \ref{Fig:beacon} represents the routing algorithm of current LN. We notice that the average hops between any two nodes of current LN is about $3.6$, which is less compared with our approach. This stems from that the route discovering algorithm of current LN does not need go through any beacons. Under some circumstances, a \textit{\textbf{shares}} issue might appear, which cannot be fulfilled by current LN whereas it might be fulfilled by Rapido. In addition, each node has a chance to be elected as a beacon, which makes all nodes possible to earn transfer fees and thus incents them to be online more often. Furthermore, Rapido can preserve the privacy of payment value without any trust function. Under estimating a transfer fee, the customer extra pays about $10$ Satoshis (sato)\footnote{$1$ bitcoin = $1 \times {10^8}$ Satoshis, $1$ Satoshi $ \approx $ $8 \times {10^{ - 5}}$\$ \cite{blockconfirm}.} fees to a node in general when a value of $1,000$ satos payment goes through a node, which is extremely small. Therefore, the quantitative difference of average hops is acceptable.

\begin{figure}
\centering
\includegraphics[scale=0.55]{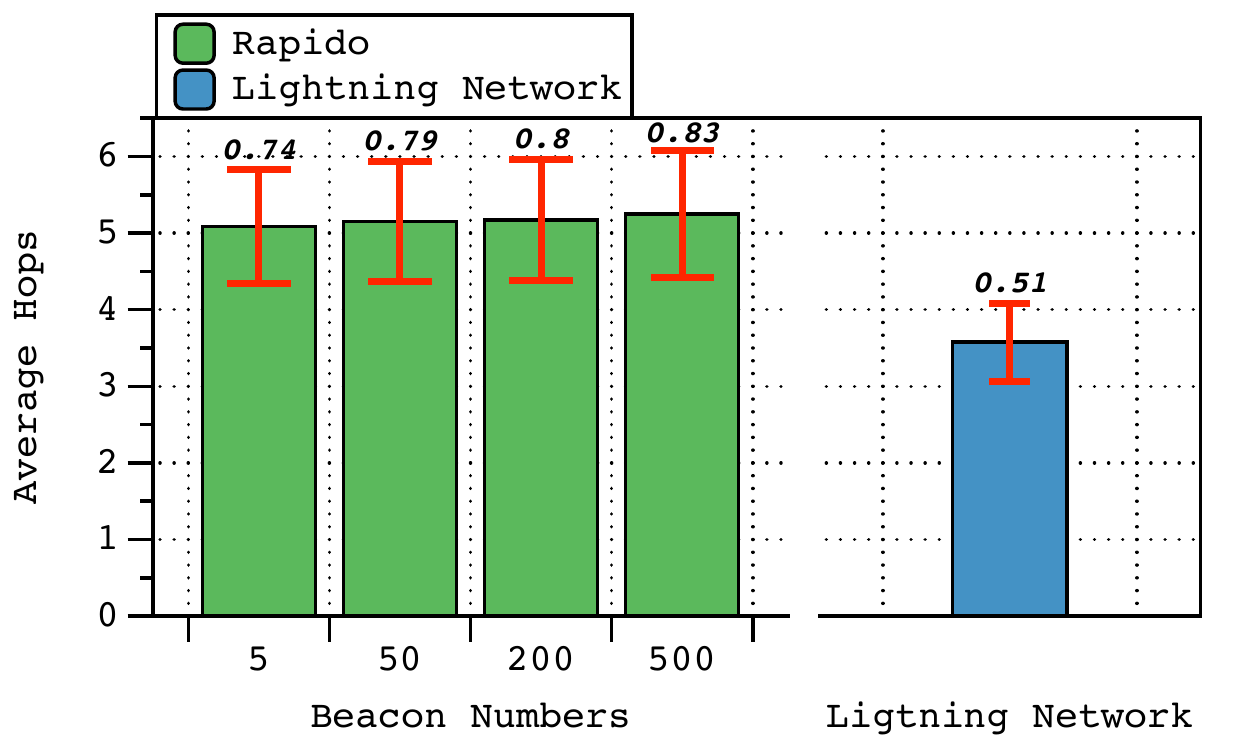}
\caption{Average hops statics between any two nodes, which is simulated on Rapdio and LN, respectively.}
\label{Fig:beacon}

\end{figure}

\subsection{Transaction Scenarios}

\subsubsection{Scenario 1}

We investigate the success rate of micropayments on LN and Rapido respectively through simulation. According to the above mentioned statics and assumptions, a series of micropayments are generated between randomly selected customer-merchant pairs (the selected customer owns enough deposits in all his payment channels to perform the payments). Considering that in a LN with $2,681$ nodes \cite{lne}, the probability of concurrent payment is pretty low, thus we does not take concurrent payments into account in our simulation.

Fig. \ref{Fig:fee} shows average transfer fees under different payment value (we empirically choose $10,000$, $25,000$, $50,000$ and $100,000$ satos according to the nature of micropyment) and different systems (yellow and green bars denote for Rapido, blue ones denote for LN). We attribute large error bars to the randomness of nodes selection strategy, in which the routing paths between selected nodes vary a lot. More than one group of fee restrictions were set for Rapido, and the results showed that adaptive restrictions ($30$, $60$, $80$, $100$) can achieve a lower average transfer fee compared with the fixed one ($200$, $200$, $200$, $200$), and achieves almost the same average transfer fee with which in LN. Fig. \ref{Fig:transaction} shows the corresponding success rate. With similar transfer fees, Rapido always achieves a high success rate than LN. For LN, the customer require to find the available shortest routing path when a transaction is generated. Additionally, there is no guarantee that an available routing path should exit for a payment. For Rapido, benefiting from the VDP program, the payment value by Rapido can be split into shares and pay to the merchant respectively. 

Furthermore, we note that Rapido can achieve a higher success rate if the restriction of transfer fee is relaxed (as the yellow bars in histogram compared with the green ones). However, we observe that $100$\% success rate is not achieved, due to the fact that once all nodes attempt to connect with all beacons, then the beacons might not enough deposits to perform the payment. In addition, combined with Fig. \ref{Fig:fee}, we can note that not large transfer fee increment (about $10\sim 160$ satos $8 \times {10^{ - 4}} \sim 1.28 \times {10^{ - 2}}$\$) could bring a large improvement on success rate. Therefore, to achieve a higher success rate, the transfer fee restriction can be relaxed.

\begin{figure}
\centering
\includegraphics[scale=0.55]{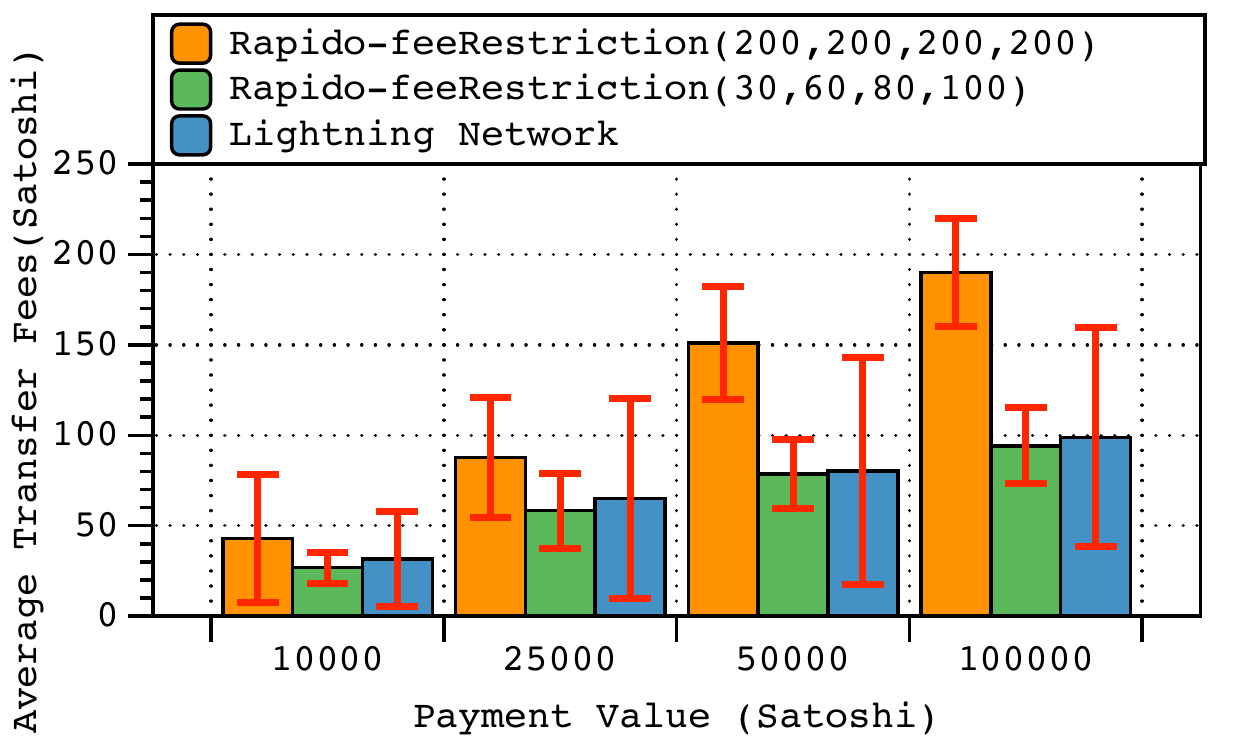}
\caption{The average and standard deviation of transfer fees based on Rapdio and LN, respectively.}
\label{Fig:fee}

\end{figure}

\begin{figure}
\centering
\includegraphics[scale=0.55]{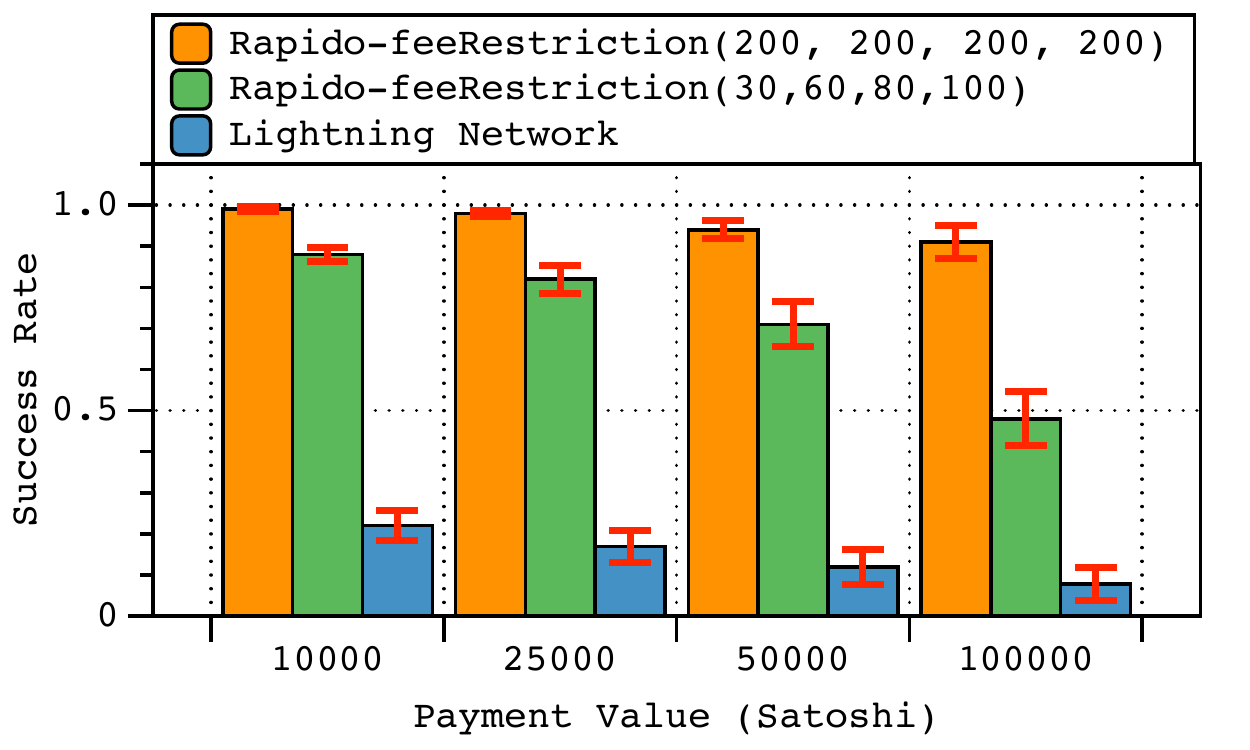}
\caption{The success rate of payments which are based on Rapdio and LN, respectively.}
\label{Fig:transaction}

\end{figure}

\subsubsection{Scenario 2}

In this scenario, we randomly select $20$ customer-merchant pairs, each of which independently generates payments of $10,000$, $25,000$, $50,000$ and $100,000$ satos respectively for $15$ times. Note that the payments can only be initiated by customers and then received by merchants. In addition, we assume that the selected customers has enough deposits to fulfill the payments. Under this assumption, both customers and merchants should be huge nodes which both have $60+$ payment channels

To better describe a skewed node, we introduce a $skewness$ metric to denote the skewness of a node in a certain direction. Concretely, a payment of value $z$ attempts to go through a node with an input channel and an output channel that the deposit named $Z_{in}$ and $Z_{out}$ respectively. Thus the skewness on the direction from input to output is ${{{Z_{out}}} \mathord{\left/
 {\vphantom {{{Z_{out}}} {{Z_{in}}}}} \right.
 \kern-\nulldelimiterspace} {{Z_{in}}}}$. Since the node earns a transfer fee of $Fee$, $skewness$ then becomes ${{\left( {{Z_{out}} - z} \right)} \mathord{\left/
 {\vphantom {{\left( {{Z_{out}} - z} \right)} {\left( {{Z_{in}} + z + Fee} \right)}}} \right.
 \kern-\nulldelimiterspace} {\left( {{Z_{in}} + z + Fee} \right)}}$ until the payment is done. Apparently, if several payments with large value goes through this node, its $skewness$ could decrease badly, which means that the node tends to be ``very skewed''. In this paper, we call a node is seriously skewed if its $skewness$ is less than $0.01$. Furthermore, if there are $n$ nodes seriously skewed among all $m$ nodes which involved in a payment, we say that the ratio of serious skewed is ${\raise0.5ex\hbox{$\scriptstyle n$}
\kern-0.1em/\kern-0.15em
\lower0.25ex\hbox{$\scriptstyle m$}}$.

Fig. \ref{Fig:skew} shows the ratio of serious skewed node after $15$ rounds of above mentioned payments from customers to merchants. We find out that there exit skewed nodes even when both sides of payments are huge nodes, and the ratio of seriously skewed nodes tends to be higher in LN than in our Rapido implementation. Our experiment proves that Rapido is capable of mitigating the skewness of nodes and congestion after several payments.

\begin{figure}
\centering
\includegraphics[scale=0.55]{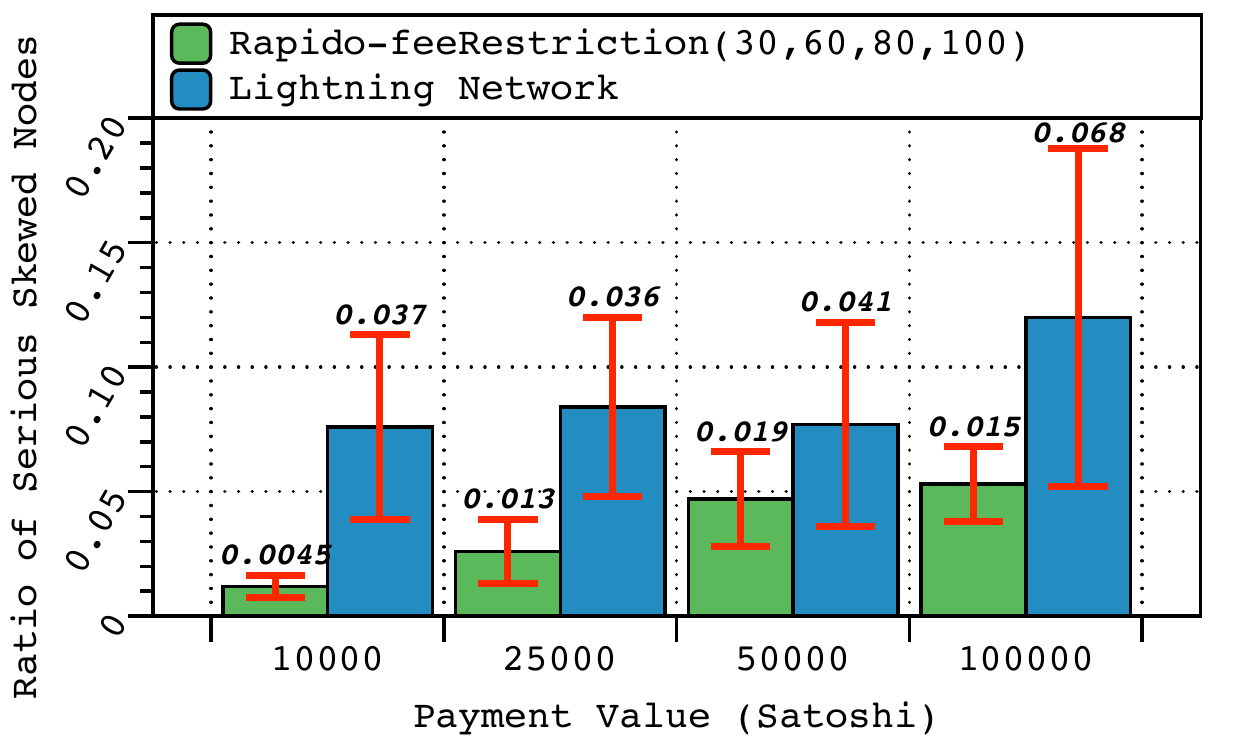}
\caption{The ratio of serious skewed nodes based on Rapdio and LN, respectively.}
\label{Fig:skew}

\end{figure}

\section{Related Work}
\label{related}

There are some limitations on LN and some prior research works have studied on them. Most of the works mainly focus on customer/merchant anonymity. Heilman proposed TumbleBit which allows parties to make fast, anonymous, off-blockchain payments through an untrusted intermediary called the Tumbler \cite{tumblebit}. Green and Miers represented Bolt that allows for secure, instantaneous and private payments that substantially reduce the storage burden on the LN \cite{bolt}. Malavolta et al. proposed Fulgor and Rayo \cite{concurrencyandprivacy} that Fulgor and Rayo are payment protocols for LN that can protect the balance security, payment value privacy and customer/merchant anonymity based on a trust function. Fulgor is a blocking protocol and therefore prone to deadlocks of concurrent payments, whereas Rayo is a protocol that enforces non-blocking progress (i.e., at least one of the concurrent payments terminates). 
In addition, some research works studied on the applicability of LN such that a skewed LN might be generated after multiple transactions. Khalil and Gervais proposed REVIVE that allows an arbitrary set of user in LN to rebalance their nodes, according to the preferences of the nodes \cite{revive}. REVIVE can only be leveraged to solve a skewed loop network, in which might cause an increased collateral cost on payment routing path.


\section{Conclusion}
\label{conclusion}
The idea of the Layer2 Network is proposed to resolve the bitcoin scalability problem, which theoretically enables fast transactions between participates. However, there exists some drawbacks. In this paper, we obeserve a new issue, \textit{\textbf{shares}} issue, which might bring some unnecessary costs. To mitigate \textit{\textbf{shares}} issue, we propose Rapido which is equipped with D-HTLC. Furthermore, Rapido also mitigates the skewness issue. In addition, Rapido inherently preserves the privacy of total payment value, for which the payment value is split into shares that an intermediate node cannot intercept the whole information. We have conducted extensive experiments to evaluate Rapido, which demonstrate that the proposed system outperforms the state-of-the-arts.

\bibliographystyle{IEEEtran}

\bibliography{Mybib}

\end{document}